\documentclass[conference, letterpaper, 10p]{IEEEtran}

\usepackage{gensymb}
\usepackage{lipsum} 
\usepackage{graphicx}
\usepackage{stfloats}
\usepackage{array}
\usepackage{mathtools}
\usepackage{amssymb}
\usepackage{amsthm,amsmath,amsfonts}
\usepackage[caption=false,font=footnotesize]{subfig}
\usepackage{float}
\usepackage{url}
\usepackage[utf8]{inputenc}
\usepackage[english]{babel}
\usepackage{mathrsfs}

\usepackage[backend=biber,
style=ext-numeric-comp,
maxbibnames=6,
sorting=none,
url=false,
doi=false,
hyperref=true,
firstinits=true,
date=year,
natbib, sorting = none]{biblatex}
\DeclareFieldFormat
 [article,inbook,incollection,inproceedings,patent,thesis,unpublished]{titlecase:title}{\MakeSentenceCase{#1}}

\newcommand{\rimax}{r_{i\max}}
\newcommand{\aimax}{a_{i\max}}
\newcommand{\uimax}{u_{i\max}}
\newcommand{\uimin}{u_{i\min}}
\newcommand{\dijmin}{d_{ij\min}}

\theoremstyle{remark}
\newtheorem{theorem}{Theorem}
\newtheorem{remark}{Remark}
\newtheorem{definition}{Definition}

\newtheorem{proposition}{Proposition}
\DeclareMathOperator{\atan}{atan2}

\IEEEoverridecommandlockouts
\bibliography{references}
 
\allowdisplaybreaks
\begin{document}
%
\title{Hybrid Control Barrier Functions for \\Nonholonomic Multi-Agent Systems }


\author{Aurora Haraldsen, Josef Matou\v{s}, and Kristin Y. Pettersen
\thanks{*This result is part of a project that has received funding from the European Research Council (ERC) under the European Union’s Horizon 2020 research and innovation programme, through the ERC Advanced Grant 101017697-CRÈME. }
\thanks{Aurora Haraldsen, Josef Matou\v{s}, and Kristin Y. Pettersen are with the Department of Engineering Cybernetics, Norwegian University of Science and Technology, N-7491 Trondheim, Norway.  {\{\tt\small aurora.haraldsen, josef.matous, kristin.y.pettersen\}@ntnu.no}}%
}


%


\maketitle

\begin{abstract}
This paper addresses the problem of guaranteeing safety of multiple coordinated agents moving in dynamic environments. It has recently been shown that this problem can be efficiently  solved through the notion of Control Barrier Functions~(CBFs). However, for nonholonomic vehicles that are required to keep positive speeds, existing CBFs lose their validity. To overcome this limitation, we propose a hybrid formulation based on synergistic CBFs~(SCBFs), which leverages a discrete switching mechanism to avoid configurations that would render the CBF invalid. Unlike existing approaches, our method ensures safety in the presence of moving obstacles and inter-agent interactions while respecting nonzero speed restrictions.  We formally analyze the feasibility of the constraints with respect to actuation limits, and the efficacy of the solution is demonstrated in simulation of a multi-agent coordination problem in the presence of moving obstacles.
\end{abstract}


%
\IEEEpeerreviewmaketitle

\section{Introduction}
The problem of coordinating multiple agents has been subject to considerable attention. One crucial aspect is ensuring that the agents do not collide, although the majority of existing work does not take this into consideration~\cite{BORRMANN201568}. There are mainly two ways of tackling this challenge. The first is to directly handle collisions within the coordination scheme, and the second is to handle collisions through a separate "backup" controller. An example of the former is for instance the hierarchical null-space-based behavioral (NSB) algorithm~\cite{NSB}, where collision avoidance is typically given the highest priority among a set of chosen tasks. However, such a tight integration significantly increases the complexity of the coordination design. For instance, in the NSB algorithm one needs to explicitly decide when to activate the collision avoidance task, and this task, upon activation, will prevent the fulfillment of any other interfering tasks. Such a setup can lead to poor performance, especially in dense environments, as designing such logic is nontrivial.

The structure proposed in this paper falls into the latter category, which has the benefit that the coordination scheme is undisturbed unless the nominal actions become unsafe, e.g., because of an obstacle.  Moreover, this separation leads to a modular control design, enabling seamless switching between different coordination protocols. The proposed solution, which can be implemented either in a centralized or decentralized fashion (we will focus on the centralized case), involves solving an optimization problem where collision avoidance is encoded as affine constraints on the input, i.e. a quadratic program~(QP). The safety constraints are derived through the newly popularized Control Barrier Functions~(CBFs), which have the desirable property of being minimally restrictive with respect to the nominal coordination objective.

The use of CBFs in multi-agent systems is not new. Notably, similar problems have been solved using CBFs in \cite{BORRMANN201568,Wang17,Wang2017,matous21,col_liv_cbf,tan2022distributed}, but for holonomic vehicles. The formulation of (valid) CBFs in the presence of nonholonomic constraints is, in fact, far less trivial because it typically  results in an undefined relative degree for certain (usually symmetric) configurations.  Some remedies have been suggested for such systems. One common approach, see e.g. \cite{robotarium,Glotfelter19}, is to map single-integrator dynamics to (nonholonomic) unicycle dynamics  using a near-identity diffeomorphism~\cite{saber}. Others have suggested to formulate the safety constraint more unconventionally on the velocities of the system~\cite{goswami2024collision} or add a heading-dependent term~\cite{molnar2021model} to the constraint.  However, these approaches are limited to vehicles that can reverse to avoid collisions, such as differential-drive robots.  In contrast, the solution proposed in this paper can safely be applied to nonholonomic vehicles that are required to maintain positive and even constant speeds. 
 
Building on the work of \cite{marley2021synergistic,Marley24}, the proposed method involves  introducing multiple CBFs such that when the relative degree becomes undefined, there exists another (valid) CBF. Robust behavior is achieved through a hybrid control design using a hysteresis-based switching logic~\cite{sanfelice2021hybrid}.  Whereas \cite{marley2021synergistic,Marley24} only considered the presence of a static obstacle, we propose an extension to the dynamic setting comprising both inter-agent collisions as well as agent-obstacle collisions. Safety guarantees are rigorously established, and we additionally analyze the feasibility of the constraints under explicit input bounds. The approach is integrated with a nominal coordination strategy, a variation of the NSB algorithm proposed in \cite{matous24}, to demonstrate its application in a highly practical scenario. We showcase the efficacy of the method in simulations of multiple nonholonomic agents moving in dynamic environments, supporting the theoretical results.

The paper is organized as follows. Preliminary concepts are established in Section~\ref{sec:prel}. Section~\ref{sec:problem} defines the control problem. Section~\ref{sec:control} outlines the proposed control method and proves the safety guarantees thereof. Section~\ref{sec:sim} demonstrates the application of the control method through numerical  simulations, and Section~\ref{sec:conclusion} gives the concluding remarks.

\section{Preliminaries}\label{sec:prel}
This section introduces the relevant notation. We also describe the use of control barrier functions (CBFs) for guaranteeing safety of continuous-time systems and the modeling framework adopted for representing systems that can evolve both in continuous and discrete time, known as hybrid systems.
\subsection{Notation} Vectors are generally denoted as $\boldsymbol{x}$, matrices as $X$, scalars as $x$, and sets as $\mathscr{X}$. The $n$-dimensional Euclidean space is denoted as $\mathbb{R}^n$, $\mathbb{R}_{>0}$ is the set of positive real numbers, and $\mathbb{R}_{\geq 0}$ is the set of nonnegative real numbers. $\mathbb{Z}$ is the set of integers and $\mathbb{Z}_{>0}$ and $\mathbb{Z}_{\geq0}$ are the sets of positive and nonnegative integers, respectively. The Euclidean norm is represented by $\Vert \! \cdot \! \Vert$ and the $1$-norm by $|\! \cdot \!|$. The operator $\preccurlyeq$ represents the componentwise inequality.  $R: \mathbb{R} \mapsto \mathrm{SO}(2)$ denotes the rotation matrix corresponding to the counterclockwise rotation about the origin of the 2-dimensional Euclidean space, and $S \triangleq  R(\frac{\pi}{2})$. The Lie derivative of a scalar function $h$ in the direction of the vector field $f$ is denoted as $L_{f}h$. For a set $\mathscr{S}$, $\partial \mathscr{S}$ is the boundary and $\mathrm{int}(\mathscr{S})$ is the interior. A function $\alpha:\mathbb{R}\mapsto \mathbb{R}$ is an extended class-$\mathscr{K}$ function if it is strictly increasing and $\alpha(0) = 0$, denoted by $\alpha \in \mathscr{K}_e$. 

\subsection{Control Barrier Functions}
A comprehensive overview of CBFs and their application to dynamical systems can be found in \cite{ames19}. The following theory is included due to its relevance for the upcoming sections. Consider the (continuous-time) affine control system
\begin{equation}\label{eq:csys}
    \dot{\boldsymbol{x}}=f(\boldsymbol{x})+g(\boldsymbol{x})\boldsymbol{u}, \quad \boldsymbol{x}\in \mathscr{X} \subset \mathbb{R}^{n_x},\quad  \boldsymbol{u}\in \mathscr{U}\subset \mathbb{R}^{n_u},
\end{equation}
where $\boldsymbol{x},\boldsymbol{u}$ are the state and input vectors, and $f : \mathscr{X}\mapsto \mathbb{R}^{n_x}$ and $g : \mathscr{X}\mapsto \mathbb{R}^{n_x\times n_u}$ are locally Lipschitz continuous. 
\begin{definition}[Forward Control Invariance~\cite{harms2025safequadrotornavigationusing}]\label{def:invariance}
    A set $\mathscr{S}$ is forward control invariant for \eqref{eq:csys} if, for every $\boldsymbol{x}\in \mathscr{S}$, there exists at least one input trajectory $\boldsymbol{u}(t) \in \mathscr{U}$ such that $\boldsymbol{x}(t) \in \mathscr{S}$ for all $t\geq t_0$. 
\end{definition}
Definition~\ref{def:invariance}  characterizes the property of set forward invariance under the influence of a control input. Suppose that $\boldsymbol{x} \in \mathscr{S}$ is analogous to safety. Then, forward (control) invariance of $\mathscr{S}$ is desirable.  More generally, we seek a set of safeguarding controls $\mathscr{U}_{\mathrm{safe}} \subset \mathscr{U}$ that evolve the system inside $\mathscr{S}$. This motivates the definition of CBFs.
\begin{definition}[Control Barrier Function~\cite{ames19}]
   Let $\mathscr{S}$ be the $0$-superlevel set of a continuously differentiable function $h:\mathscr{X}\mapsto \mathbb{R}$, then $h$ is a CBF for \eqref{eq:csys} if there exists an extended class-$\mathscr{K}$ function $\alpha$ such that,   for all $\boldsymbol{x} \in \mathscr{S}$,
   \begin{equation}\label{eq:h_dot_alpha_def}
   \underset{\boldsymbol{u} \in \mathscr{U}}{\sup}     \dot{h}(\boldsymbol{x},\boldsymbol{u}) \geq -\alpha(h(\boldsymbol{x})).
   \end{equation}
\end{definition}
Indeed, the constraint \eqref{eq:h_dot_alpha_def} is both sufficient and necessary for invariance of the set $\mathscr{S}$, or \emph{safety}, thus it is minimally restrictive \cite{ames19}. Moreover, for control-affine systems, as in \eqref{eq:csys}, the constraint is additionally linear in the control, as seen from
\begin{equation}
    \dot{h}(\boldsymbol{x},\boldsymbol{u}) \equiv L_{f}h(\boldsymbol{x})+L_{g}h(\boldsymbol{x})\boldsymbol{u}.
\end{equation}
\subsection{Hybrid Systems}
A hybrid dynamical system is modeled as the inclusion~\cite{sanfelice2021hybrid}:
\begin{equation}
    \mathscr{H} = \begin{cases}
        \,\ \dot{\boldsymbol{x}} \in \mathscr{F}(\boldsymbol{x}),&\boldsymbol{x} \in \mathscr{C},\\
        \boldsymbol{x}^+ \!\in \mathscr{G}(\boldsymbol{x}), & \boldsymbol{x} \in \mathscr{D},
    \end{cases}
\end{equation}
representing that the state $\boldsymbol{x}$ can evolve both in continuous time, referred to as flow, and in discrete time, referred to as jumps.  Accordingly, we refer to $\mathscr{F}, \mathscr{G}$ as the flow and jump maps, respectively. Moreover, $\mathscr{C}$ is the flow set and $\mathscr{D}$ is the jump set. Well-posedness and robustness of solutions are guaranteed under the hybrid basic conditions \cite[Assumption~6.5]{sanfelice2021hybrid}. CBFs for hybrid dynamical systems and conditions for forward invariance thereof are established in \cite{MAGHENEM2021109328,Marley24}. These results are omitted due to space considerations.
\section{Problem Description}\label{sec:problem}
We consider a setting comprising $n_a \in \mathbb{Z}_{> 0}$ cooperating (nonholonomic) agents, each associated with an index $i \in \mathscr{I}_a$,  modeled as
    \begin{align}\label{eq:model}
     \begin{bmatrix}
           \dot{x}_i\\
           \dot{y}_i\\
           \dot{\psi}_i\\
           \dot{u}_i
       \end{bmatrix} = \underbrace{ \begin{bmatrix}
         u_i \cos\psi_i\\
         u_i \sin\psi_i\\
         0 \\
         0
       \end{bmatrix}}_{f_i(\boldsymbol{x}_i)}        + \underbrace{\begin{bmatrix}
           0 & 0 \\
           0 & 0 \\
           1 & 0\\
           0 & 1
       \end{bmatrix}}_{g_i(\boldsymbol{x}_i)}\begin{bmatrix}
           r_i\\
           a_i
       \end{bmatrix},
    \end{align}
where $\boldsymbol{p}_i \triangleq [x_i,y_i]^\top$ and $\boldsymbol{v}_i \triangleq \dot{\boldsymbol{p}}_i$ represent the position and velocity of the $i$th agent, respectively, and $u_i \in \mathbb{R}$ and $\psi_i \in (-\pi,\pi]$ are the agent's speed and orientation. Furthermore, $\boldsymbol{x}_{i} \triangleq [\boldsymbol{p}^\top_i,\psi_{i},u_i]^\top\in \mathscr{X}_i$ is the state vector, with $\mathscr{X}_i \triangleq \mathbb{R}^2\times (-\pi,\pi]\times \mathbb{R}$, and $\boldsymbol{u}_i \triangleq [r_i, a_i]^\top \in \mathscr{U}_i\subset \mathbb{R}^2$ is the input vector.  The control inputs comprise the heading rate and acceleration $r_i \in \mathscr{R}_i$ and $a_i \in \mathscr{A}_i$, respectively, and the set of feasible heading rates and accelerations are $\mathscr{R}_i\triangleq \{r_i \in \mathbb{R} :  | r_i| \leq \rimax\}$ and $\mathscr{A}_i\triangleq \{a_i \in \mathbb{R}:| a_i| \leq \aimax \}$, where $\rimax > 0$ and $\aimax \geq 0$ are constant parameters. 

The model \eqref{eq:model} is chosen for its wide applicability to vehicles found on land, sea, and in the air. The agents are assumed to move with strictly positive speeds, i.e., we restrict $u_i$ to the interval $[\uimin, \uimax]$, where $\uimax\geq \uimin > 0$. This restriction is common in real-world systems, for example in fixed-wing aircrafts and marine vehicles during maneuvering. Note that this implies that $u_i \equiv \Vert  \boldsymbol{v}_i\Vert$.  The goal is to coordinate the agents along a desired path \emph{safely}, i.e. without collisions between the agents and external obstacles. These two aspects are described separately below.

\begin{itemize}
    \item \textbf{Formation Path-Following:} This task represents the nominal control objective, and it involves coordinating the fleet of agents along a desired path. To define the task more precisely, we introduce the so-called barycenter of the formation:
    \begin{equation}\label{eq:pb}
        \boldsymbol{p}_b \triangleq \frac{1}{n_a}\sum_{i \in \mathscr{I}_a} \boldsymbol{p}_i.
    \end{equation}
   The goal is to control the vehicles such that their barycenter converges  to the desired path. For simplicity, the path is chosen to coincide with the $x$-axis of the inertial coordinate frame. The path-following task is then
   \begin{equation}\label{eq:obj_path}
      \lim_{t\rightarrow \infty} y_b(t)= 0.
   \end{equation}
    In addition, the vehicles should keep a prescribed formation relative to the barycenter (and, ultimately, the path).  The desired formation is described by $n_a$ position vectors $\boldsymbol{p}_1^f,\dots, \boldsymbol{p}^f_{n_a}$, chosen such that $\sum_{i\in\mathscr{I}_a}\boldsymbol{p}_i^f = \boldsymbol{0}$. Note that, in general, the formation positions should be rotated according to the orientation of the path. This is not necessary in our case, since the path we chose is aligned with the inertial frame. We will formulate this objective using the task variables $ \boldsymbol{\sigma}_{i,f} \triangleq \boldsymbol{p}_i- \boldsymbol{p}_b$,  representing the position of the $i$th vehicle expressed in the path-tangential frame~\cite{matous24}. Defining the corresponding task errors as $\tilde{ \boldsymbol{\sigma}}_{f,i} \triangleq  \boldsymbol{\sigma}_{f,i}- \boldsymbol{\sigma}_{f,i,d}$ with $\boldsymbol{\sigma}_{f,i,d} \triangleq \boldsymbol{p}_i^f$, the goal of the formation-keeping task can be expressed as
    \begin{equation}\label{eq:obj_form}
     \lim_{t\rightarrow \infty} \sum_{i \in \mathscr{I}_a} \left\Vert\tilde{\boldsymbol{ \sigma}}_{f,i} (t)\right\Vert = 0.
    \end{equation}
    \item \textbf{Collision Avoidance:} 
The fleet of vehicles may encounter obstacles along the path that require evasive action. Given $n_o \in \mathbb{Z}_{\geq 0}$ obstacles, we will denote the position and velocity of the $j$th obstacle as $\boldsymbol{p}_j \triangleq [x_j,y_j]^\top$ and $\boldsymbol{v}_j \triangleq \dot{\boldsymbol{p}}_j$, respectively, associated with the indices $j \in \mathscr{I}_o$ such that $\mathscr{I}_o \cap \mathscr{I}_a = \emptyset$. Denote the pairwise distances as $d_{ij}= d_{ji} \triangleq \Vert \boldsymbol{p}_{ij} \Vert, \boldsymbol{p}_{ij}\triangleq \boldsymbol{p}_j-\boldsymbol{p}_i$, and suppose that there is a collision if this distance is reduced below a prescribed limit of $\dijmin > 0$. We thus require that
       \begin{align}\label{eq:obj_obscolav}
       \begin{split}
           \inf_{t\geq 0} d_{ij}(t)\geq \dijmin,\quad\forall i \in \mathscr{I}_a,\quad \forall j \in \mathscr{I}_o.
       \end{split}
    \end{align}
Agents must also maintain a prescribed distance to each other.  Hence, we similarly impose that
    \begin{align}\label{eq:obj_ivcolav}
     \inf_{t\geq 0} d_{ij}(t) \geq \dijmin,\,\, \forall i,j \in \mathscr{I}_a, \,\,j \neq i.
    \end{align}
\end{itemize}

The NSB algorithm \cite{NSB} is used for the nominal control design. Since there are no new contributions directly related to this framework,  the details of it are skipped. Please refer to \cite{matous24} for a rigorous treatment of the method and details on the control design specific to the tasks \eqref{eq:obj_path} and \eqref{eq:obj_form}.

\begin{remark}
   Collision avoidance has traditionally been implemented as high-priority tasks within the NSB framework~\cite{NSB}. This design choice poses several challenges. First, tasks are analogous to equality constraints, whereas collision avoidance is fundamentally an inequality constraint. Although this limitation can be circumvented, another more crucial aspect is that the resulting output velocity does not take into account underlying constraints and may be infeasible~\cite{matous24}. The proposed approach does not suffer from such drawbacks and is straightforward to combine with most coordination algorithms.
\end{remark}

\section{Safety Guarantees  for Nonholonomic Agents using Hybrid Control Barrier Functions}\label{sec:control}
The control structure proposed in this work implements collision avoidance on an external layer on top of the nominal control. We follow a standard min-norm control design using control barrier functions (CBFs) to encode safety. A benefit of this approach is that the controller is independent of the underlying nominal algorithm and can thus be paired with a broad range of multi-agent behaviors. Different from related approaches is that we only require that each agent controls its heading rate~$r_i$ (presuming positive speed), and we do not resort to backup controllers as in \cite{Squires2018}.

We propose to compute the heading rates of the cooperating agents in a centralized manner:
\begin{equation}\label{eq:QP}
    \begin{split}
        & \underset{\boldsymbol{r} \in \mathscr{R}}{\arg \min}    \quad \frac{1}{2} \Vert \boldsymbol{r} - \boldsymbol{r}^d\Vert^2,\\
        &\text{subject to} \quad  G \boldsymbol{r} \preccurlyeq \boldsymbol{b},
    \end{split}
\end{equation}
where $\boldsymbol{r} \triangleq [r_1, r_2, \dots , r_{n_a}]^\top$ and $\boldsymbol{r}^d \triangleq [r^d_1, r^d_2, \dots , r^d_{n_a}]^\top$ are vectors of (desired) control inputs, $\mathscr{R} \triangleq \mathscr{R}_1 \times \mathscr{R}_2 \times \cdots \times \mathscr{R}_{n_a}$, and $ G\in \mathbb{R}^{n_c \times n_a}$ and $\boldsymbol{b} \in \mathbb{R}^{n_c}$, where $n_c = n_an_o +\frac{n_a(n_a-1)}{2}$ represents the number of linear constraints.  The acceleration input is simply set to the desired acceleration given by the nominal controller, i.e. $a_i = a_i^d$ for all $i \in \mathscr{I}_a$.

\begin{remark}
 The controller \eqref{eq:QP} can also be implemented in a decentralized setting, for instance by performing the optimization only for the closest neighboring agents or even by treating all other agents as obstacles, the latter corresponding to a single-agent system.
\end{remark}

This remainder of the section formulates the constraint matrices for the optimization~\eqref{eq:QP}, structured into three parts. Sections~\ref{ssec:obsavoid} and \ref{ssec:agent_avoid} derive the constraints for preventing agent-obstacle and inter-agent collisions, respectively. In Section~\ref{ssec:coop_ob_avoid}, we present an alternative formulation of the former which is particularity relevant to formation control problems. It consists of coordinating the agents such that they keep the formation during evasive maneuvers. This can be beneficial for communication purposes and in other cases where it is necessary to keep a tight formation e.g. whilst collectively transporting an object.

\subsection{Agent-to-Obstacle Constraint}\label{ssec:obsavoid}
To fulfill the pairwise objectives \eqref{eq:obj_obscolav} involving vehicle $i$ and obstacle $j$, we define functions
\begin{equation}\label{eq:h0}
    h_{0}(\boldsymbol{x})= {\dijmin}^2-{d_{ij}}^2.
\end{equation}
For ease of exposition, we will leave out the index terms when it is clear from the context. For instance, we define the stacked vector $\boldsymbol{x} \triangleq [\boldsymbol{x}^\top_i, \boldsymbol{x}^\top_{j}]^\top$ and the set $\mathscr{X}\triangleq \mathscr{X}_i\times \mathscr{X}_{j} \backslash \mathscr{B}_{ij}$, where $\boldsymbol{x}_{j}\in \mathscr{X}_{j} \subset \mathbb{R}^{n_{j}}$ is the state of the obstacle, which evolves through some function $f_j: \mathscr{X}_{j} \mapsto  \mathbb{R}^{n_{j}}$.
Moreover,  $\mathscr{B}_{ij}$ defines the ball $\mathscr{B}_{ij} \triangleq \{\boldsymbol{x}\in \mathscr{X}_i\times\mathscr{X}_j: d_{ij} < \epsilon \}$, where $\epsilon < \dijmin$ is an arbitrarily small positive constant. Note that the usual sign is reversed, i.e. $h_0\leq 0$ for each $i,j$ implies safety, and so $\mathscr{S}_0 \triangleq \left\{\boldsymbol{x}\in \mathscr{X}: h_0(\boldsymbol{x})\leq 0\right\}$. The derivative of a function $h_*$ along the system \eqref{eq:model} will be expressed using Lie derivatives:
\begin{equation}\label{eq:hastdot}
    \dot{h}_* = L_{f}h_* +\left(L_{g}h_*\right)_2a_i+\left(L_{g}h_*\right)_1 r_i,
\end{equation}
such that $\left(L_{g}h_*\right)_m$ represents the $m$th element of the vector $L_{g_i} h_* \in \mathbb{R}^2$ and $L_{f}h_* \triangleq L_{f_{ij}}h_*$, where $f_{ij} \triangleq [f_i^\top, f_j^\top]^\top$ are the stacked vector fields.

Given the higher relative degree, $\left(L_{g}h_0\right)_1(\boldsymbol{x})= 0$ for all $\boldsymbol{x}\in \mathscr{X}$, we introduce the following extension based on~\cite{xiao2021high}:
\begin{equation}\label{eq:h1}
   h_1(\boldsymbol{x})= \dot{h}_0(\boldsymbol{x})+ \alpha_0\left(h_0(\boldsymbol{x})\right), \quad\dot{h}_0(\boldsymbol{x}) = 2\boldsymbol{p}^\top_{ij}\boldsymbol{v}_{ij},
\end{equation}
where $\boldsymbol{v}_{ij} \triangleq \boldsymbol{v}_i-\boldsymbol{v}_{j}$ and $\alpha_0\in \mathscr{K}_e$. The function $h_1$ is a candidate high-order CBF (HOCBF) of order 2. Further,  $h_1(\boldsymbol{x}) \leq 0$ implies $\dot{h}_0(\boldsymbol{x}) \leq -\alpha_0(h(\boldsymbol{x}))$ for all $\boldsymbol{x}\in \mathscr{X}$ and 
\begin{align}\label{eq:Lgh1}
     \left(L_{g}h_1\right)_1(\boldsymbol{x}) =  2 u_i \boldsymbol{p}_{ij}^\top  S\boldsymbol{n}_{i},
\end{align}
where we denote $\boldsymbol{n}_i = \frac{\boldsymbol{v}_{i}}{\Vert \boldsymbol{v}_i\Vert}$ for the unit vector in the direction of $\boldsymbol{v}_i$. However, solving $ \left(L_{g}h_1\right)_1(\boldsymbol{x}) = 0$ shows that the influence of the heading rate vanishes for all $\boldsymbol{x} \in\mathscr{X}$ satisfying $\psi_i = \atan(y_{ij},x_{ij})$, i.e. the relative degree becomes undefined. A similar issue has been addressed in \cite{marley2021synergistic,Marley24} for static obstacles.
\begin{remark}
Note that the term $\left(L_g h_1\right)_2$ is nonzero in the above case. Nevertheless, a hybrid control design is needed for the following reasons. For vehicles that require positive speed, safety cannot be maintained solely through acceleration control, as this would require negative speed~\cite{Haraldsen24ACC}. Even if the vehicle can reverse to avoid a collision, it may enter a deadlock by doing so~\cite{Marley24}. Integration of the vehicle's acceleration in the present strategy is, however, a subject for future work.
\end{remark}
Inspired by \cite{marley2021synergistic,Marley24}, the proposed solution builds on hybrid systems theory, and the idea is similar as for synergistic control using Lyapunov functions  \cite{sanfelice2021hybrid}. It consists of designing a family of CBFs, such that when the relative degree becomes undefined, there exists another CBF with a lower value, as illustrated in Figure~\ref{fig:synergycbf}.

To avoid points with vanishing gradient, we introduce a discrete variable $q \in \mathscr{Q}$ that can take the values $\mathscr{Q} \triangleq \left\{-1,1\right\}\subset \mathbb{Z}$.  The resulting dynamics can be modeled as the hybrid system 
\begin{align}\label{eq:hybrid_sys}
    \mathscr{H} = \begin{cases}
        \quad \! \!\!\dot{\boldsymbol{x}}\in \mathscr{F}(\boldsymbol{x}, q),\quad \,\,\,\, \dot{q}= 0,&(\boldsymbol{x},q) \in \mathscr{C},\\
      q^+ \in \mathscr{G}\,(\boldsymbol{x},q), \quad \boldsymbol{x}^+ = \boldsymbol{x},& (\boldsymbol{x},q) \in \mathscr{D},
    \end{cases}
\end{align}
where the flow and jump sets $\mathscr{C}, \mathscr{D} \subset \mathscr{X}\times \mathscr{Q}$ and the jump map $\mathscr{G}\subset \mathscr{Q}$ will be defined subsequently, and the flow map is
\begin{equation}\label{eq:flow}
    \mathscr{F} = \left\{ \,\begin{bmatrix}
        f_i(\boldsymbol{x}_i) +g_i(\boldsymbol{x}_i)\boldsymbol{u}_i\\
        f_j(\boldsymbol{x}_{j})
    \end{bmatrix} : \boldsymbol{u}_i\in \mathscr{U}_i  \,\right\}.
\end{equation}
Similar to \cite{Marley24}, the logical variable is used to shift the principal rotations by an angle $\vartheta \in \mathbb{R}_{>0}$ in clockwise or counterclockwise direction. We propose the following extension of \eqref{eq:h1}:
\begin{equation}\label{eq:h2}
\begin{split}
    \!\!\!\! h_2(\boldsymbol{x},q) = 2\boldsymbol{p}_{ij}^\top R\left(q \vartheta\right)\boldsymbol{v}_{i}\! -\!  2\boldsymbol{p}_{ij}^\top\boldsymbol{v}_{j}\! +\!  \alpha_0(h_0(\boldsymbol{x})) + \varepsilon(\boldsymbol{x}), 
\end{split}
\end{equation}
where $\varepsilon(\boldsymbol{x}) \triangleq  4 d_{ij} u_i \sin\left(\vartheta\right)$ is a penalty term. The function~\eqref{eq:h2} essentially represents two candidate CBFs  or, equivalently, one candidate synergistic CBF~(SCBF) \cite{marley2021synergistic}. The term $\varepsilon$ is needed to ensure that forward invariance of the set 
\begin{equation}
    \mathscr{S}_2 \triangleq \left\{(\boldsymbol{x},q)\in \mathscr{X} \times \mathscr{Q} : h_2(\boldsymbol{x},q)\leq 0\right\}
\end{equation} 
implies safety, as elaborated next.
\begin{figure}
    \centering
    \includegraphics[width=1\linewidth]{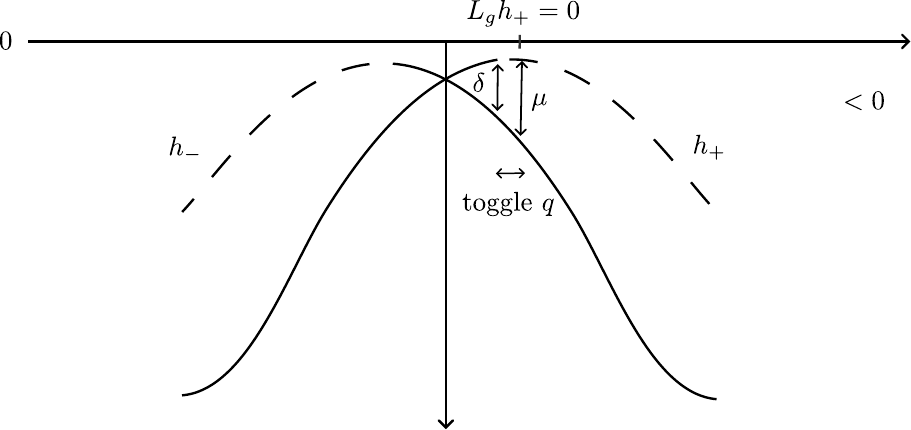}
    \caption{Illustration of a synergistic CBF, comprising two candidate CBFs toggled by a discrete variable. }
    \label{fig:synergycbf}
\end{figure}

\begin{proposition}\label{prop:obstacle_safety}
 For the system \eqref{eq:hybrid_sys}-\eqref{eq:flow}, if $\vartheta < \frac{\pi}{2}$ and $\mathscr{S}_2$ is forward invariant, then $\mathscr{S}_0$ is forward invariant.
\end{proposition}
\begin{proof}
    The proof follows from showing that $ \mathscr{S}_2 \subset \mathscr{S}_1$, with $\mathscr{S}_1 \triangleq\left\{ \boldsymbol{x} \in \mathscr{X} : h_1(\boldsymbol{x})\leq 0\right\}$. Manipulation of \eqref{eq:h1} gives
    \begin{align}
    h_1(\boldsymbol{x}) = h_2(\boldsymbol{x},q)+ 2\boldsymbol{p}_{ij}^\top \Delta(q\vartheta) \boldsymbol{v}_i-\varepsilon(\boldsymbol{x}),
\end{align}
where $ \Delta(q\vartheta) \triangleq  I - R(q\vartheta)$ and $ I$ is the 2-by-2 identity matrix. Furthermore,
\begin{align}\label{eq:Deltabound}
     \boldsymbol{p}_{ij}^\top  \Delta(q\vartheta)\boldsymbol{v}_{i} \leq  d_{ij} u_{i}\Vert  \Delta \Vert \leq 2  d_{ij} u_i\sin\left(\vartheta\right).
\end{align}
It follows that $h_1(\boldsymbol{x})\leq h_2(\boldsymbol{x},q),\forall \boldsymbol{x},q \in \mathscr{X} \times \mathscr{Q}$. This fact implies that $\mathscr{S}_0$ is forward invariant (see e.g. \cite{xiao2021high} for a proof). 
\end{proof}
The strategy is to update $q$ such that we avoid the critical points. The critical set is defined as the set of states where the influence of the heading rate vanishes and the function $h_2$ is nondecreasing:
\begin{equation}
\begin{split}
  \mathscr{Z} \triangleq \left\{ (\boldsymbol{x},q,a_i) \in  \mathscr{X} \times \mathscr{Q} \times \mathscr{A}_i :  L_{f}h_2(\boldsymbol{x},q) + \right.&\\
 \left. \left(L_{g}h_2\right)_2(\boldsymbol{x},q)a_i\geq 0, \,\left(L_g h_2\right)_1(\boldsymbol{x},q)  =0\,\right\}&.
\end{split}
\end{equation}
For each $q = \pm1$, the critical orientations corresponding to $\left(L_g h_2\right)_1(\boldsymbol{x},q) = 0$ are given by
\begin{equation}\label{eq:crit}
 \psi_{i}  = \atan(y_{ij},x_{ij})\mp \vartheta.
\end{equation}
In order for $h_2$ to be a candidate SCBF, the function $h_2$ must strictly decrease during jumps. This property is satisfied if $h_2$ has a nonzero synergy gap~\cite{marley2021synergistic}.  To define this measure, let $m_2(\boldsymbol{x}) \triangleq \underset{q \in \mathscr{Q}}{\min} \,h_2(\boldsymbol{x},q)$. Then, the synergy gap is given by 
\begin{equation}\label{eq:mu}
    \mu =\underset{\boldsymbol{x},q\in \mathscr{Z}}{\inf}\Big[h_2(\boldsymbol{x},q) - m_2(\boldsymbol{x})\Big].
\end{equation}

\begin{proposition}\label{prop:obstacle_SCBF}
    If $\vartheta < \frac{\pi}{2}$, the function \eqref{eq:h2} is a candidate SCBF for \eqref{eq:hybrid_sys}-\eqref{eq:flow}. That is, the synergy gap $\mu$ in \eqref{eq:mu} is nonzero.
\end{proposition}
\begin{proof}
Due to symmetry, we only need to show that this holds for the critical point $\mathscr{Z}^+ = \{(\boldsymbol{x},q)\in \mathscr{Z}: q = +1\}$. It is easy to verify that $\underset{q \in \mathscr{Q}}{\arg \min} \, h_2 (\boldsymbol{x},q) = -1$ for all $\boldsymbol{x} \in \mathscr{Z}^+$.  The synergy gap is thus given by
 \begin{equation}\label{eq:mu_proof}
    \mu = \underset{\boldsymbol{x} \in \mathscr{Z}^+}{\inf} 2 u_{i} d_{ij} \boldsymbol{n}_{ij}^\top \left(  R(\vartheta) -  R^\top(\vartheta) \right)\boldsymbol{n}_i,
 \end{equation}
 where  $\boldsymbol{n}_{ij}\triangleq \frac{\boldsymbol{p}_{ij}}{\Vert \boldsymbol{p}_{ij}\Vert}$ is the unit vector in the direction of $\boldsymbol{p}_{ij}$. 
Moreover, since $\boldsymbol{n}_{ij}^\top R(\vartheta) \boldsymbol{n}_{i} = 1$ and $\boldsymbol{n}_{ij}^\top   R^\top(\vartheta) \boldsymbol{n}_{i} < 1$ for all $\boldsymbol{x} \in \mathscr{Z}^+$ by \eqref{eq:crit}, the synergy gap is positive.  
\end{proof}
Through some manipulation of \eqref{eq:mu_proof}, we can express the synergy gap as
\begin{equation}
    \mu = \underset{\boldsymbol{x}\in \mathscr{X}}{\inf} 2 u_i  d_{ij} (1-\cos(2\vartheta)),
\end{equation}
which is positive for all $\vartheta < \frac{\pi}{2}$. Note that if $\mathscr{X}$ contained the ball $\mathscr{B}_{ij}$, the synergy gap would theoretically be zero. However, since $\boldsymbol{x} \in \mathscr{B}_{ij}$ implies that the vehicle is already inside the obstacle, excluding these states does not pose an issue in practice. 

The results of Propositions~\ref{prop:obstacle_safety}-\ref{prop:obstacle_SCBF} motivate the following update law for $q$, expressed through the hybrid system~\eqref{eq:hybrid_sys}:
\begin{subequations}
    \begin{align}\label{eq:G_set}
        \mathscr{G}(\boldsymbol{x}) \quad\,\,=\quad\,\, \left\{q \in \mathscr{Q} : h_2(\boldsymbol{x},q) - m_2(\boldsymbol{x}) = 0 \right\},\\\label{eq:D_set}
      \mathscr{D} = \{ (\boldsymbol{x},q) \in \mathscr{X} \times\mathscr{Q}:h_2(\boldsymbol{x},q) - m_2(\boldsymbol{x})\geq \delta\},\\
      \mathscr{C} = \{ (\boldsymbol{x},q) \in \mathscr{X} \times\mathscr{Q}:h_2(\boldsymbol{x},q) - m_2(\boldsymbol{x})\leq \delta\},\label{eq:C_set}
    \end{align}
\end{subequations}
where $\delta > 0$ is the desired hysteresis width, which can be state-dependent such that $\delta \coloneqq \delta(\boldsymbol{x})$. Importantly, by selecting $\delta < \mu$, solutions are not allowed to flow near the critical points, and $h_2$ is strictly decreasing during jumps. This is illustrated in the example shown in Figure~\ref{fig:synergycbf}.

 For every agent-obstacle pair $i,j$, we thus incorporate the constraint:
\begin{equation}\begin{split}\label{eq:h2_constraint}
\!\!\!\dot{h}_2(\boldsymbol{x},q, \boldsymbol{u}_i)\leq -\alpha_2(h_{2}(\boldsymbol{x},q))\!\Leftrightarrow \!
    g_{ij}(\boldsymbol{x},q)r_i \leq  b_{ij}(\boldsymbol{x},q),
\end{split}
\end{equation}
where $b_{ij}(\boldsymbol{x},q) =-L_f h_2(\boldsymbol{x},q) -\left(L_g h_2\right)_2(\boldsymbol{x},q)a_i^d-\alpha_2(h_2(\boldsymbol{x},q))$ and $g_{ij}(\boldsymbol{x},q) = \left(L_g h_2\right)_1(\boldsymbol{x},q)$.  Compactly,
\begin{equation}
\begin{split}\label{eq:Go}
   \boldsymbol{g}_i& = \begin{bsmallmatrix}
        g_{i1}&
        g_{i1}& \cdots & 
       g_{in_o}
    \end{bsmallmatrix}^\top, \,\,\,
        \boldsymbol{b}_i = \begin{bsmallmatrix}
   b_{i1}&
        b_{i1}& \cdots &
        b_{in_o}
    \end{bsmallmatrix}^\top, \\
    \quad  G' & =  \begin{bsmallmatrix}
        \boldsymbol{g}_{1}&  \boldsymbol{0} & \cdots& \cdots & \boldsymbol{0} \\
        \boldsymbol{0} & \boldsymbol{g}_2 &  \boldsymbol{0} &\cdots &  \boldsymbol{0}\\
        \vdots & \ddots & \ddots & \ddots & \vdots \\
        \boldsymbol{0} & \cdots & \boldsymbol{0} & \boldsymbol{g}_{n_a-1} &  \boldsymbol{0}\\
        \boldsymbol{0} & \cdots & \cdots &  \boldsymbol{0} & \boldsymbol{g}_{n_a}
    \end{bsmallmatrix}, \quad \boldsymbol{b}' = \begin{bsmallmatrix}
        \boldsymbol{b}_1\\
        \boldsymbol{b}_2 \\
        \vdots \\
        \boldsymbol{b}_{n_a-1} \\
        \boldsymbol{b}_{n_a}
    \end{bsmallmatrix}.
\end{split}
\end{equation}
The matrix $G'\in \mathbb{R}^{n_on_a\times n_a}$ and vector $\boldsymbol{b}'\in \mathbb{R}^{n_on_a}$ represent the collection of agent-to-obstacle constraints to be included in \eqref{eq:QP}. The remaining constraints are formulated next.
\subsection{Agent-to-Agent Constraint}\label{ssec:agent_avoid}
The derivation of the constraints for pairwise safety between agents is almost equivalent to the previous derivation. The main difference is that the vehicles can cooperate to avoid potential collisions. We will continue with the same notation as above, i.e. $\boldsymbol{x} = [\boldsymbol{x}_i, \boldsymbol{x}_j]^\top$ are the stacked states of the $i$th and $j$th agent etc. Moreover, for each $i,j \in \mathscr{I}_a$ with $j \neq i$, let $h_0$ and $h_1$ be defined as \eqref{eq:h0} and \eqref{eq:h1}, respectively. We will denote the derivative of a function $h_*$ as in \eqref{eq:hastdot} augmented with the control inputs of agent $j$:
\begin{equation}\label{eq:h2_dot_a}
    \dot{h}_* = L_{f}h_* + \sum_{k=i,j} \left(L_{g_k}h_*\right)_2 a_k+ \left(L_{g_k}h_*\right)_1 r_k.
\end{equation}
 
The function $h_1$ has undefined relative degree with respect to $r_i,r_j$ when the orientation of the relative velocity vector $\boldsymbol{v}_{ij}$ coincides with the orientation of the relative position vector~$\boldsymbol{p}_{ij}$ such that $\atan(\dot{y}_{ij},\dot{x}_{ij}) =  \atan(y_{ij},x_{ij})$, and the agents' headings align with that same line. This occurs for instance when the agents are moving directly toward each other's centers, which  motivates a control design essentially identical to the one used for agent-obstacle safety. 

Specifically, we define the function $h_2$ equivalent to \eqref{eq:h2}, but an equally valid option would be to shift the principal rotations with respect to agent $j$ in this case. The resulting hybrid system is on the form \eqref{eq:hybrid_sys},\eqref{eq:G_set}-\eqref{eq:C_set} with flow set
\begin{equation}\label{eq:flow_agents}
    \mathscr{F} = \left\{ \,\begin{bmatrix}
        f_i(\boldsymbol{x}_i) +g_i(\boldsymbol{x}_i)\boldsymbol{u}_i\\
        f_j(\boldsymbol{x}_j) +g_j(\boldsymbol{x}_j)\boldsymbol{u}_j
    \end{bmatrix} : \boldsymbol{u}_i\in \mathscr{U}_i, \boldsymbol{u}_j\in \mathscr{U}_j  \,\right\}.
\end{equation}
Similar to \eqref{eq:crit}, the critical orientations correspond to
\begin{align}
    \psi_{i}= \atan(y_{ij},x_{ij}) \mp \vartheta,\quad \psi_j = \atan(y_{ij},x_{ij})
\end{align}
and the vector $ R(\pm \vartheta) \boldsymbol{v}_i-\boldsymbol{v}_j$ oriented along the line $\boldsymbol{p}_{ij}$.The proofs of the following two propositions are omitted because they follow the same reasoning as in the proofs of Propositions~\ref{prop:obstacle_safety}-\ref{prop:obstacle_SCBF}.

\begin{proposition}\label{prop:agents_safety}
 For the system \eqref{eq:hybrid_sys},\eqref{eq:flow_agents}, if $\vartheta < \frac{\pi}{2}$ and $\mathscr{S}_2$ is forward invariant, then $\mathscr{S}_0$ is forward invariant.
\end{proposition}
\begin{proposition}\label{prop:agents_SCBF}
    If $\vartheta < \frac{\pi}{2}$, the function \eqref{eq:h2} is a candidate SCBF for \eqref{eq:hybrid_sys},\eqref{eq:flow_agents}. That is, the synergy gap $\mu$ in \eqref{eq:mu} is nonzero.
\end{proposition}
With $g_i^j(\boldsymbol{x},q) =\left(L_{g_i}h_2\right)_1(\boldsymbol{x},q)$, $g_j^i(\boldsymbol{x}) = \left(L_{g_j}h_2\right)_1(\boldsymbol{x})$, and
$b_{ij}(\boldsymbol{x},q) = -L_f h_2(\boldsymbol{x},q)- \underset{k=i,j}{\sum} \left(L_{g_k}h_2\right)_2(\boldsymbol{x},q) a_k^d - \alpha_2(h_2(\boldsymbol{x},q))$, the resulting constraint matrices are given by
\begin{align}\label{eq:Ga}
\begin{split}
    & G''= \begin{bsmallmatrix}
     g_{1}^2 & g_{2}^1 & 0\quad & 0\quad &\cdots &  0 & 0 \\
     g_{1}^3 & 0 & g_{3}^1 & 0 \quad& \cdots &0 & 0 \\
     & & && \vdots& &  & \\
     0 & 0 & 0\quad & 0\quad& \cdots& g_{n_a}^{n_a-1} &\mathclap{\,\,\,g_{n_{a}-1 }^{n_a}},
 \end{bsmallmatrix}, \,\,\boldsymbol{b}'' =  
 \begin{bsmallmatrix}
         b_{12} \\ 
         b_{13} \\ \\ \vdots \\  b_{n_a-1 n_a } 
     \end{bsmallmatrix}
 \end{split},
\end{align}
where the  row dimension (i.e., the number of constraints) equals $\frac{(n_a-1)n_a}{2}$. 

By Proposition~\ref{prop:obstacle_SCBF} and \ref{prop:agents_SCBF}, the function \eqref{eq:h2} is a valid SCBF when agent $i$ has an unbounded input space, i.e. $\mathscr{R}_i = \mathbb{R}$. However, we usually have $\rimax < \infty$. If $j$ represents a static obstacle, safety can nevertheless be assured under certain conditions on the functions $\alpha_0,\alpha_2$ \cite{Marley24}.

These guarantees do not hold when accounting for the dynamics of a moving obstacle. Still, the existence of feasible control actions that guarantee safety can be facilitated by a constructive choice of the parameters involved in the constraint \eqref{eq:h2_constraint}, as shown in the following theorem.

\begin{theorem}
Suppose that  $\underset{t\geq0}{\sup} \left\Vert\boldsymbol{v}_{j}(t)\right\Vert\coloneqq u_{o\max}<\infty$ and $\underset{t\geq0}{\sup}  \left\Vert \dot{\boldsymbol{v}}_{j}(t)\right\Vert\coloneqq a_{o\max} <\infty $. If
 \begin{equation}\label{eq:vartheta_bound}
      \sup_{\mathclap{\substack{\\[1mm]\vartheta \in (0,\frac{\pi}{2})\\\gamma > 0}}}  \left[4\beta_1\sin^2(\vartheta) - \left( \sqrt{5} \beta_1 + \beta_2\right )  \sin(\vartheta) - \beta_3\right] > 0,
    \end{equation}
   with $\beta_1 \triangleq \rimax$, $\beta_2\triangleq 2\left(\frac{\uimax+u_{o\max}}{\dijmin}\right)+\frac{\aimax}{\uimin}$, and $\beta_3 \triangleq \frac{a_{o\max}+a_{i\max}}{\uimin} + \gamma\left(\frac{u_{o\max}+u_{i\min}}{\uimin}\right)$, then  $\mathscr{S}_0$ is forward (control) invariant for the system~\eqref{eq:flow}.
\end{theorem}
\begin{proof}
We show that $h_2$ is increasing along the system~\eqref{eq:hybrid_sys}-\eqref{eq:flow} on the boundary $\partial \mathscr{S}_2 = \{(\boldsymbol{x},q)\in\mathscr{X}\times \mathscr{Q}: h_2(\boldsymbol{x},q) = 0\}$, for explicitly bounded inputs, which implies, by Proposition \ref{prop:obstacle_safety}, that $\mathscr{S}_0$ is forward (control) invariant. For this to hold, we require that 
\begin{equation}\label{eq:r_i_bound}
\begin{split}
     \sup_{r_i \in \mathscr{R}_i}  - \left(L_{g}h_2\right)_1(\boldsymbol{x},q) r_i >&\\ \left(L_{g}h_2\right)_1(\boldsymbol{x},q)a_i + L_fh_2(\boldsymbol{x},q) &,
\end{split}
\end{equation}
for all $(\boldsymbol{x},q) \in \partial \mathscr{S}_2$ and $a_i \in\mathscr{A}_i$. This inequality is trivially satisfied if the right-hand side is negative. Therefore, a sufficient condition for \eqref{eq:r_i_bound} is 
\begin{equation}\label{eq:r_i_bound_strict}
\begin{split}
      \underset{
    (\boldsymbol{x},q)\in\partial \mathscr{S}_2}{\inf}\,\Vert\left(L_{g}h_2\right)_1(\boldsymbol{x},q)\Vert \rimax >\\ \underset{\substack{(\boldsymbol{x},q)\in \partial \mathscr{S}_2\\ a_i \in \mathscr{A}_i}}{\sup} \left[\left(L_{g}h_2\right)_1(\boldsymbol{x},q)a_i + L_fh_2(\boldsymbol{x},q)\right].
\end{split}
\end{equation}
It can be established from \eqref{eq:h2} that
\begin{align}
    \begin{split}
        L_fh_2(\boldsymbol{x},q) &\leq 2d_{ij}\Big(\nabla \alpha_0(h(\boldsymbol{x})) \Vert \boldsymbol{v}_{ij}\Vert+\Vert\dot{\boldsymbol{v}}_{j}\Vert\Big)+  \\  &2\Vert  \boldsymbol{v}_{ij}\Vert u_i\Big( \Vert  \Delta(q,\boldsymbol{x})\Vert+\sin\left(\vartheta\right)\Big),
    \end{split}
      \\
   \left(L_g h_2\right)_2 (\boldsymbol{x},q) &\leq 2d_{ij}\Big(1+\sin\left(\vartheta\right)\Big),\,\, \forall (\boldsymbol{x},q) \in \partial \mathscr{S}_2.
\end{align}
Evaluating \eqref{eq:h2},\eqref{eq:mu_proof} under the update law \eqref{eq:G_set}-\eqref{eq:C_set} gives
\begin{align}
\left(L_g h_2\right)_1 (\boldsymbol{x},q) &\geq \delta - 2d_{ij}u_{i}\left\Vert \Lambda(q\vartheta)\right\Vert,
\end{align}
for all $(\boldsymbol{x},q) \in \partial \mathscr{S}_2$, where  $ \Lambda(q\vartheta) \triangleq  R(q \vartheta)- R^\top(q\vartheta) -  S R(q\vartheta)$. Taking $\alpha_0(h) = \gamma h, \gamma>0$ and using that $\delta < \mu$  and $\Vert \Lambda(q\vartheta) \Vert \leq \sqrt{5} \sin(\vartheta)$,  it follows that \eqref{eq:r_i_bound_strict} is guaranteed if there exists $\vartheta < \frac{\pi}{2}$ such that
\begin{equation}\begin{split}\label{eq:r_bound_strictstrict}
     \left(4\sin^2(\vartheta) - \sqrt{5}\sin(\vartheta)  \right)\rimax >\gamma\left( \frac{\Vert \boldsymbol{v}_{ij}\Vert}{u_i}\right)+\\ 2\left(\frac{\Vert\boldsymbol{v}_{ij}\Vert}{d_{ij}} +\frac{\aimax}{2 u_i}\right) \sin(\vartheta)+\frac{a_{o\max} + \aimax }{u_i},
\end{split}
\end{equation}
    where we have substituted $1-\cos(2\vartheta) = 2\sin^2(\vartheta)$. Using $\Vert \boldsymbol{v}_{ij}\Vert \leq\Vert \boldsymbol{v}_{i}\Vert+\Vert \boldsymbol{v}_{j}\Vert$ and rearranging \eqref{eq:r_bound_strictstrict} results in \eqref{eq:vartheta_bound}.
\end{proof}
In the case of cooperating agents~\eqref{eq:flow_agents}, note that the constraint regressor of agent $j$ (i.e. the term $\left(L_{g_j}h_{2}\right)_1(\boldsymbol{x})$ in \eqref{eq:h2_dot_a}) is independent of the discrete variable $q$. Therefore, the hybrid solution does not prevent this term from becoming zero. Consequently, the analysis of this system is equivalent to the above and is thus omitted.

\begin{remark}
   Our solution cannot strictly guarantee safety for orientations $\psi_{i} = \atan(y_{ij},x_{ij})\mp \vartheta + \pi$.    However, this situation implies that agent $i$ is oriented away from obstacle (resp. agent)  $j$. Hence, such configurations can only be contained in the critical set $\mathscr{Z}$ if the obstacle (resp. agent) is moving toward the vehicle and at a higher speed than it (and if $j$ is an agent, $\psi_j = \atan(y_{ij},x_{ij})$). We have neglected these edge cases since safety cannot be guaranteed by heading control alone in such scenarios. However, safety can be guaranteed by integrating acceleration control. For example, if $j$ represents a collision-seeking obstacle, increasing the vehicle's speed above that of the obstacle would suffice to avoid it. If $j$ is another agent, they could cooperate to resolve the situation, e.g. agent~$j$ could reduce its speed while agent~$i$ increases its own. Since these solutions respect the nonzero speed constraint and do not require modifications to the controller \eqref{eq:QP}, we may assume that appropriate accelerations are applied if any of the described situations arise. Explicit utilization of the acceleration input will be considered in the future.
\end{remark}

\subsection{Collaborative Agents-to-Obstacle Constraint}\label{ssec:coop_ob_avoid}
  A more unified evasive behavior can be induced by treating the fleet's barycenter, defined in~\eqref{eq:pb}, as the collision point, thereby controlling the formation as a single unit, similar to \cite{matous24}. To formulate this constraint, we need to express the safety distance between the barycenter and the obstacle in a way that reflects agent safety.  Denote $d_{bi} \triangleq \Vert \boldsymbol{\sigma}_{f,i}\Vert$ for the distance between the barycenter $\boldsymbol{p}_b$ and agent $i$, expressed through the task variable introduced in Section~\ref{sec:problem}. Next, define  \begin{equation}
    d_{f} \triangleq \max_{i\in\mathscr{I}_a} d_{bi} , \quad  d_{f,d}\triangleq
     \max_{i\in\mathscr{I}_a} \Vert \boldsymbol{\sigma}_{f,i,d} \Vert,
\end{equation}
such that $d_{f}$ is a (varying) radius encompassing every vehicle in the formation, and $d_{f,d}$ is the (constant) desired formation radius. Then, $\underset{t\geq 0}{\inf} d_{bj}(t) \geq d_{bj\min}$ implies \eqref{eq:obj_obscolav} if we take 
\begin{equation}\label{eq:dbjmin}
\begin{split}
  d_{bj\min} \triangleq\max_{i\in\mathscr{I}_a} d_{ij\min} +\max\left\{d_{f}, d_{f,d} \right\}.
\end{split}
\end{equation}

The formulation of this variable, particularly the second term, is motivated by the fact that, if the agents strive to keep their formation, then the formation radius will only grow when it is smaller than the desired formation radius. It is therefore safe to treat $d_{bj\min}$ as constant or nonincreasing, which simplifies the control design. Note, however, that this simplification is only appropriate when the nominal velocities are designed as in \cite{matous24} or formation-keeping is enforced in other ways, e.g., as an additional constraint in~\eqref{eq:QP}. 

\begin{remark}
Involving the derivative of \eqref{eq:dbjmin} is problematic since this function is nonsmooth. Furthermore, doing so would introduce the formation radius as a property that can be varied to preserve safety, which is undesirable. An alternative is to supplement the barycenter-based constraint with the pairwise constraints \eqref{eq:Go} for additional safeguarding in case agents significantly deviate from the prescribed formation. 
\end{remark}

The function $h_0$ representing collaborative collision avoidance of an obstacle $j$ is defined as
\begin{equation}\label{eq:h0_b} 
    h_0(\boldsymbol{x}) = {d_{bj\min}}^2-{d_{bj}}^2, 
\end{equation}
where $\boldsymbol{x} = [\boldsymbol{x}_1^\top, \dots,\boldsymbol{x}^\top_{n_a}, \boldsymbol{x}^\top_{j}]^\top$. Since a nonhybrid formulation of the form \eqref{eq:h1} would result in undefined relative degree when all agents are oriented in the direction of the line $\boldsymbol{p}_{bj}$, we propose the following candidate SCBF:
\begin{equation}\label{eq:h2_b}
  h_2(\boldsymbol{x},q) =2 \boldsymbol{p}_{bj}^\top R(q \vartheta) \boldsymbol{v}_{b} -2\boldsymbol{p}_{bj}^\top \boldsymbol{v}_{j}+\alpha_0(h_0(\boldsymbol{x}))+\varepsilon(\boldsymbol{x}),
\end{equation}
where the penalty term is $\varepsilon(\boldsymbol{x}) = 4 u_b d_{bj} \sin(\vartheta)$, where $u_b \triangleq \Vert \boldsymbol{v}_b\Vert\equiv  \Vert \frac{1}{n_a}\sum_{i\in\mathscr{I}_a} \boldsymbol{v}_i\Vert$ is the speed of the barycenter. It is straightforward to verify that safety can be ensured through \eqref{eq:h2_b} given that $u_b\neq 0$, i.e., the synergy gap is nonzero. Note, however, that forward invariance of $\mathscr{S}_0$ depends on whether $d_{bj\min}$ is in fact nonincreasing. 
 
 The following constraints can replace the pairwise constraints given by \eqref{eq:Go} or they can be jointly included in \eqref{eq:QP}:
\begin{equation}
\begin{split}\label{eq:Gof}
    \boldsymbol{g}_{bj}& = \frac{2\boldsymbol{p}_{bj}^\top R(q\vartheta)  S}{n_a}\begin{bmatrix}\boldsymbol{v}_1&\boldsymbol{v}_2& \cdots & \boldsymbol{v}_{n_a}
    \end{bmatrix},\\
    \quad  G'''& =  \begin{bsmallmatrix}
        \boldsymbol{g}_{b1} \\
        \boldsymbol{g}_{b2} \\
        \vdots \\
        \boldsymbol{g}_{b n_o}
    \end{bsmallmatrix}, \quad \boldsymbol{b}'''= \begin{bsmallmatrix}
        b_1\\
        b_2 \\
        \vdots \\
       b_{n_o}
    \end{bsmallmatrix},
\end{split}
\end{equation}
where $b_j(\boldsymbol{x},q) = -L_f h_2(\boldsymbol{x},q) -\sum_{i\in\mathscr{I}_a}\left( L_{g_i} h_2 \right)_2 (\boldsymbol{x},q)a_i^d - \alpha_2(h_2(\boldsymbol{x},q))$ and $f = [f^\top_1, \dots,f^\top_{n_a}, f^\top_j]^\top$.  Note that $ G'''\in \mathbb{R}^{n_o\times n_a}$ and $\boldsymbol{b}'''\in \mathbb{R}^{n_o}$, so the dimension of this constraint set is reduced compared to the pairwise constraints \eqref{eq:Go}.

\section{Numerical Simulations}\label{sec:sim}
The effectiveness of the approach is demonstrated in simulations of a fleet of nonholonomic agents modeled by \eqref{eq:model}. Each agent is restricted by the bounds $\rimax = 0.5$~rad/s, $\uimin = 0.3$~m/s, $\uimax = 0.8$~m/s, and $\aimax = 0.25$~m/s$^2$. For all agent pairs, we selected the safety distances as $\dijmin = 10$~m. In the computation of the constraints defined by \eqref{eq:Go} and \eqref{eq:Ga}, we used $\vartheta = \frac{\pi}{4}$~rad, the functions $\alpha_0,\alpha_2$ were taken as linear with $\gamma_0 = 1,\gamma_2 = 0.1$, and the hysteresis width was selected as $\delta(u_i,d_{ij}) = u_id_{ij}(1-\cos(2\vartheta))$. The same design was used for the barycenter-based constraints~\eqref{eq:Gof}, apart from the hysteresis width which was taken as $\delta(u_b, d_{bj}) = u_b d_{bj}(1-\cos(2\vartheta))$. The nominal controller was designed according to \cite{matous24}.

In the first example, we simulated five vehicles, with desired formation shape given by 
\begin{equation}
    \begin{bmatrix}
        \boldsymbol{p}^f_1, .., \boldsymbol{p}^f_5
    \end{bmatrix} = \begin{bmatrix}
        30 & -30 & 0 & -30  & 30 \\
        -20 & 30 & 0 & -30  & 20
    \end{bmatrix}.
\end{equation}
 Two obstacles were simulated to approach the fleet, represented by indices $6$ and $7$, and we used the safety radii $d_{i6\min}= 15$~m and $d_{i7\min} = 12$~m for all agents $i = 1,\dots,5$. The optimization \eqref{eq:QP} was solved with $G = [G'^\top,G''^\top]^\top$ and $\boldsymbol{b}= [\boldsymbol{b}'^\top,\boldsymbol{b}''^\top]^\top$. The vehicles were initialized away from their formation positions, as shown in Figure~\ref{fig:sim1}, creating several head-on situations between the agents. Whereas the nonhybrid formulation~\eqref{eq:h1} would in these cases result in undefined relative degree,  the proposed approach enables the vehicles to converge safely to their formation positions. Later, the vehicles deviate temporarily from their formation positions to avoid collisions with the obstacles approaching head-on and crossing. Figure \ref{fig:hmin_1} shows that the minimum value 
$h_0^{\min} \triangleq \min \left\{ \quad \min_{i,j \in\mathscr{I}_a,j\neq i} h_0, \quad \min_{i\in\mathscr{I}_a,j \in \mathscr{I}_o}h_0 \right\}$
is consistently nonnegative, verifying that safety is preserved throughout the simulation, and the fleet's barycenter deviates minimally from the desired path. Figures~\ref{fig:r_1}-\ref{fig:u_1} verify that the agents' speeds and turning rates remain within the specified bounds.

\begin{figure}[t]
    \begin{minipage}{1\linewidth}
    \subfloat[]{
    \includegraphics[width=1\linewidth, trim = 0cm 0.4cm 0cm 1.5cm, clip]{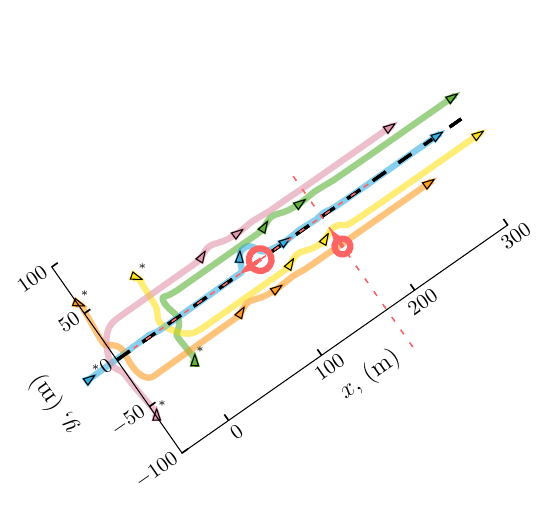}\label{fig:sim1}}
    \end{minipage}
        \begin{minipage}{1\linewidth}
            \subfloat[]{\includegraphics[width=1\linewidth]{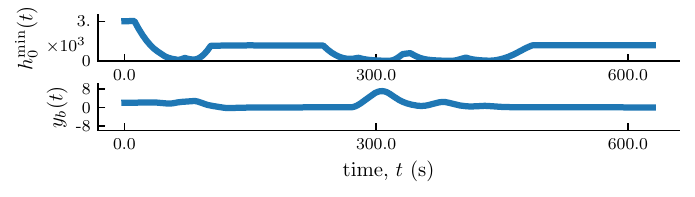}\label{fig:hmin_1}}\\
            \subfloat[]{\includegraphics[width=0.48\linewidth]{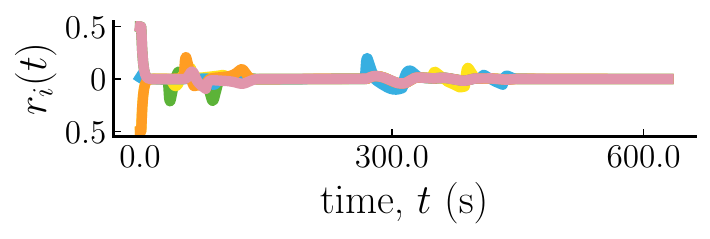}\label{fig:r_1}}%
             \subfloat[]{\includegraphics[width=0.48\linewidth]{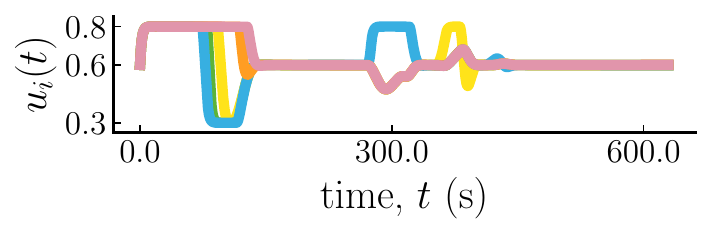}\label{fig:u_1}}
                 \caption{Results from the first simulation. The planar trajectories are shown in Figure~(a), where the vehicles are illustrated by triangular shapes and the obstacles by circular shapes. The desired path is the black, dashed line.  Figure~(b) shows the minimum $h^{\min}_0$ and the cross-track error $y_b$, and Figures~(c)-(d) show the turning rates~$r_i$ and speeds $u_i$ of all agents.}\label{fig:sim1_states}
        \end{minipage}
    
    \end{figure}
\begin{figure}[t]
    \begin{minipage}{1\linewidth}
    \subfloat[]{
    \includegraphics[width=1\linewidth, trim = 0cm 0.2cm 0cm 1.5cm, clip]{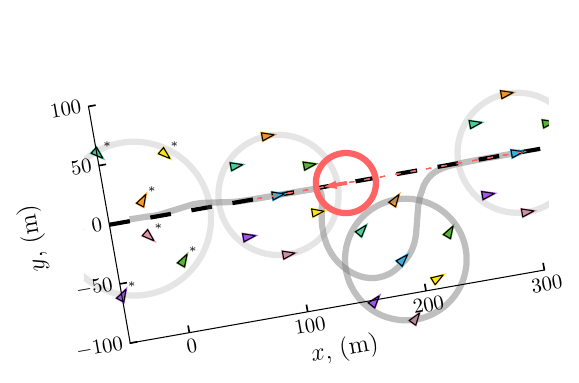}\label{fig:sim2}}
    \end{minipage}
        \begin{minipage}{1\linewidth}
            \subfloat[]{\includegraphics[width=1\linewidth]{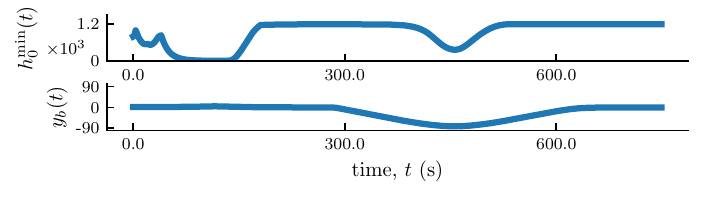}\label{fig:hmin_2}}        \\
            \subfloat[]{\includegraphics[width=0.48\linewidth]{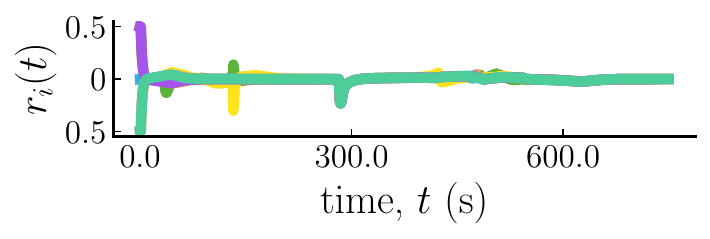}\label{fig:r_2}}
             \subfloat[]{\includegraphics[width=0.48\linewidth]{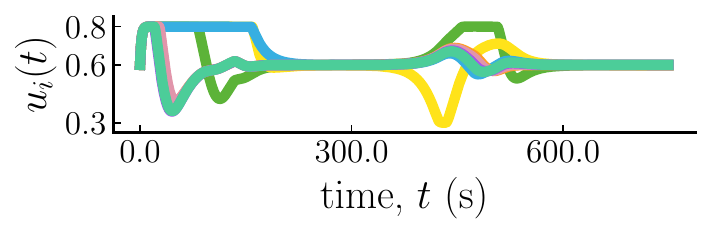}\label{fig:u_2}}
               \caption{Results from the second simulation. In Figure~(a), the gray, solid line and circles indicate the trajectory of the barycenter $\boldsymbol{p}_b$ and the formation radii, respectively. The least transparent circle shows the position of the vehicles relative to the obstacle in red.} \label{fig:sim2_states}
        \end{minipage}
    
    \end{figure}

In the second example, we consider a fleet of $n_a = 7$ vehicles, and the desired shape of the formation  was set to
\begin{equation}
\begin{split}
    \!\!\! \begin{bmatrix}
     \boldsymbol{p}^f_1, .., \boldsymbol{p}^f_7
    \end{bmatrix} \! = \! \begin{bmatrix}
        30 & -30 & 50 & 0 &-50 &-30  & 30 \\
        -20 & 30 & -50 & 0 & 50 & -30  & 20
    \end{bmatrix}
\end{split}
\end{equation}
The optimization \eqref{eq:QP} was in this case solved with $G = [G'''^\top,G''^\top]^\top$ and $\boldsymbol{b} = [\boldsymbol{b}'''^\top,\boldsymbol{b}''^\top]^\top$. We simulated a single obstacle, represented by index 8, with the safety radii $d_{i8\min} = 30$~m for all agents $i = 1,\dots,7$. The vehicles can be seen to converge to the desired formation in Figure~\ref{fig:sim2}, after which they encounter the obstacle head-on. As remarked before, this situation would have resulted in undefined relative degree if a nonhybrid formulation of the form \eqref{eq:h1} was deployed. Under the proposed strategy, however, the vehicles jointly leave the desired path, avoiding the collision whilst preserving the formation. Figures~\ref{fig:hmin_2}-\ref{fig:u_2} show that both the input/state bounds and the objectives~\eqref{eq:obj_obscolav}-\eqref{eq:obj_ivcolav} remain satisfied. Naturally, the barycenter of the formation deviates significantly more from the desired path in this scenario compared to the first scenario, which is expected given that formation-keeping was explicitly enforced via the barycenter-based constraint \eqref{eq:Gof}.

\section{Conclusions and Future Work}\label{sec:conclusion}
This paper presented a hybrid safety-critical control framework for coordinating multiple nonholonomic agents in dynamic environments, based on synergistic control barrier functions (SCBFs).  The paper addresses a key limitation of existing CBFs, which lose validity when agents are constrained to maintain strictly positive speeds, a common requirement in systems such as fixed-wing UAVs and marine vehicles.  Our contribution lies in a hybrid formulation of the CBF, involving a discrete variable that is updated such that the points where the CBF becomes invalid are avoided.  We proved the validity of the proposed SCBFs, covering both inter-agent and agent-obstacle collisions, and we analyzed the feasibility of the resulting constraints with respect to actuation limits. 

The proposed controller offers a minimally restrictive, modular safety layer than can complement a wide range of coordination schemes.  The method was integrated with a dual formation-keeping and path-tracking task, and a collaborative formation-level safety constraint based on the barycenter of the formation was proposed. We demonstrated the efficacy of the solution in simulations involving multiple nonholonomic agents and dynamic obstacles, confirming that the safety constraints remained feasible.


Future work will extend the framework to systems with more complex dynamics, such as autonomous underwater vehicles (AUVs), and implement the method in real-world experiments to further validate its practical applicability.

\printbibliography

\end{document}